\documentclass[12pt]{article}

\usepackage[margin=1.125in]{geometry}
\usepackage{xcolor}
\definecolor{DarkGreen}{rgb}{0.1,0.5,0.1}
\definecolor{DarkRed}{rgb}{0.5,0.1,0.1}
\definecolor{DarkBlue}{rgb}{0.1,0.1,0.5}

\usepackage[small]{caption}
\usepackage[pdftex]{hyperref}
\hypersetup{
    unicode=false,          
    pdftoolbar=true,        
    pdfmenubar=true,        
    pdffitwindow=false,      
    pdfnewwindow=true,      
    colorlinks=true,       
    linkcolor=DarkBlue,          
    citecolor=DarkGreen,        
    filecolor=DarkGreen,      
    urlcolor=DarkBlue,          
    pdftitle={},
    pdfauthor={},    
}
\usepackage{graphicx} 
\usepackage{amssymb,amsmath,amsthm,amsfonts,enumitem}
\usepackage{fullpage,nicefrac,comment}
\usepackage{tikz}
\usetikzlibrary{arrows,decorations.pathmorphing,decorations.shapes,snakes,patterns,positioning}
\usepackage[ruled,linesnumbered]{algorithm2e}

\newcommand{\ourcode}{Hermitian-Lifted Code}
\newcommand{\ourcodes}{Hermitian-Lifted Codes}

\newcommand{\cC}{\ensuremath{\mathcal{C}}}

\newcommand{\F}{{\mathbb F}}

\newcommand{\inset}[1]{\left\{#1\right\}}

\newcommand{\inparen}[1]{\left(#1\right)}

\newcommand{\suchthat}{\,:\,}

\newcommand{\supp}{\mathrm{Supp}}

\newcommand{\bin}{\mathrm{bin}}

\renewcommand{\Im}{\ensuremath{\operatorname{Im}}}

\renewcommand{\epsilon}{\varepsilon}

\newtheorem{theorem}{Theorem} 
\newtheorem{lemma}[theorem]{Lemma} 
\newtheorem{definition}{Definition}

\newtheorem{observation}[theorem]{Observation}

\newtheorem{remark}{Remark}
\newtheorem{claim}[theorem]{Claim}
\newtheorem{proposition}[theorem]{Proposition}

\newtheorem{fact}[theorem]{Fact}

\newtheorem{example}{Example}

\title{\ourcodes}

\author{
Hiram H. L\'{o}pez\thanks{ Department of Mathematics and Statistics, Cleveland State University.},
Beth Malmskog\thanks{Department of Mathematics and Computer Science, Colorado College.},
Gretchen Matthews\thanks{Department of Mathematics, Virginia Tech. Work partially supported by NSF DMS-1855136},  \\
Fernando Pi\~{n}ero-Gonz\'{a}lez\thanks{Department of Mathematics, University of Puerto Rico at Ponce.}, 
and Mary Wootters\thanks{Departments of Computer Science and Electrical Engineering, Stanford University.  Work partially supported by NSF CAREER award CCF-1844628, by CCF-BSF-1814629, and by a Sloan Research Fellowship.}
}
\date{}
\begin{document}

\maketitle

\begin{abstract}
In this paper, we construct codes for local recovery of erasures with high availability and constant-bounded rate from the Hermitian curve. These new codes, called Hermitian-lifted codes, are evaluation codes with evaluation set being the set of $\F_{q^2}$-rational points on the affine curve. The novelty is in terms of the functions to be evaluated;  they are a special set of monomials which restrict to low degree polynomials on lines intersected with the Hermitian curve. As a result, the positions corresponding to points on any line through a given point act as a recovery set for the position corresponding to that point. 
\end{abstract}


\section{Introduction}\label{sec:intro}

Let $\cC \subset \F^n$ be a linear code
of length $n$ over a finite field $\F$.  
For a coordinate $i \in [n]$, we say that a set $R \subseteq [n] \setminus \{i\}$ is a \em recovery set \em for the index $i$ in the code $\cC$ if the $i$'th symbol $c_i$ of a codeword $c \in \cC$ can be recovered from the symbols $\inset{ c_j \suchthat j \in R}$.  
We say that $\cC$ has \em locality \em $r$ and \em availability \em $t$ if for each $i \in [n]$, there are $t$ disjoint repair sets for $i$ in $\cC$, each of size at most $r$.  

Constructing codes with locality and availability is desirable for several reasons. Codes with extremely large availability, $t = \Omega(n)$, are known to be equivalent to \em locally decodable codes \em (LDCs)~\cite{KT00,Woo10}, which are objects of interest in theoretical computer science, complexity theory, and cryptography; see \cite{yek_survey} for a survey.
Codes with smaller availability, $t = O(1)$, have been studied recently in the context of \em distributed storage: \em if data is encoded and distributed over  multiple nodes in a distributed system, then a small piece of data can be accessed efficiently by many users simultaneously.   There are also variants on codes with locality and availability, such as \em batch codes \em and \em PIR codes, \em with applications in cryptography and private information retrieval; we refer the reader to the survey~\cite{Ska16} for more details.
Finally, codes with intermediate availability---where $t$ is sublinear in $n$ but still growing---have been studied as a bridge between the two settings above, and as an interesting problem in itself~\cite{FischerGW17,LW18,PV19}.

In this work we introduce a new  type of \em lifted code. Lifted codes are \em a class of codes which have 
given rise to 
codes with good locality and intermediate availability.  Lifted codes, introduced by Guo, Kopparty and Sudan~\cite{GuoKS13}, are evaluation codes of multivariate polynomials over large fields.  The lift of a univariate evaluation code $\cC_0$ to $m$ variables is the evaluation code corresponding to the set of all $m$-variate polynomials whose restriction to every line corresponds to a codeword in $\cC_0$.  For example, the \em lifted Reed-Solomon code \em is the code corresponding to all $m$-variate polynomials whose restriction to every line is a low-degree univariate polynomial.  Surprisingly, the set of all such polynomials can be much larger than the corresponding Reed-Muller code (corresponding to multivariate polynomials of small total degree) when the characteristic of the base field is small, and can even have rate approaching one~\cite{GuoKS13}. Several variants of the lifting operation have been proposed, with the goal of making the construction more flexible, for example: working with different base codes~\cite{GuoKS13,BGKKS13,G15}; including derivative information about the polynomials~\cite{Wu15, LW18}; and only restricting to certain sets of lines~\cite{FischerGW17}. 

In this paper we introduce a novel variant on the lifted code construction, by considering evaluation codes not on $\left(\F_q\right)^n$, but rather on the rational points of the Hermitian curve $\mathcal{H}_q$ over $\F_{q^2}$.  That is, our code corresponds to all bivariate polynomials, evaluated on $\mathcal{H}_q$, so that the restriction to any line agrees with some low-degree univariate polynomial \em on the points of $\mathcal{H}_q$ intersected with that line. \em  We call such a code a \em \ourcode, \em because we are taking the lift with respect to a Hermitian curve.

For any even prime power $q$, \ourcodes\  have length $q^3$ with locality $q$ and availability $q^2-1$.  Just as with the example of lifted Reed-Solomon codes mentioned above---and perhaps just as surprisingly---these codes have rate much larger than one would expect.  More precisely, it is not hard to see that the code described above contains 
the one-point Hermitian code $C_{q,q^2 - 1}$ (see below for notation); but this one-point code has rate that tends to $0$ as $q$ tends to infinity.  Our main mathematical result is that in fact \ourcodes\  have rate bounded below by a positive constant \em independent \em of $q$. 

Evaluation codes on the points of geometric objects offer an elegant and flexible way to create codes with good parameters, bounded through geometry, with locality arising naturally from algebraic and geometric relationships.  The rich structure of certain curves and their associated function fields enable a wide variety of geometric and algebraic perspectives that might create excellent codes.  Our construction combines the extremely effective curve-centered approach that began in \cite{Goppa} and extended to locally recoverable codes in \cite{BargTamoVladut, HaymakerMalmskogMatthews} with the lifting perspective of \cite{GuoKS13} to obtain novel codes which are not special cases of either approach.

In summary, our contributions are as follows.
First, we introduce the notion of curve-lifted codes.  This is a novel approach that combines ideas from curve-based codes and lifted codes in order to obtain good locally recoverable codes. Specifically, the evaluation points are simply rational points on the curve, as in the one-point code case, while the collection of functions to be evaluated is expanded to obtain a better code rate while guaranteeing locality and availability. We instantiate this idea by studying \ourcodes.
While \ourcodes\ do offer good locality and availability, the quantitative parameters are not better than the existing state-of-the-art.  Rather, we view the primary contribution of this work as introducing a new paradigm:
we view our construction as a proof of concept that combining these two views can result in novel codes, providing an important addition to the literature and introducing an approach that could lead to new insights and may eventually improve the state-of-the-art.

Second, we provide a positive lower bound on the rate of \ourcodes\ as $q\to \infty$.  Such a bound is surprising, since the corresponding one-point Hermitian code has a rate that tends to zero as $q \to \infty$.  Our approach is to study the set of ``good'' monomials whose restriction to any line, intersected with the Hermitian curve, agrees with a low-degree polynomial.  We give a sufficient condition for a monomial to be good, and establish via a counting argument that there are many such monomials.  This establishes a lower bound on the dimension of the code.




For the rest of the Introduction, we review the Hermitian curve and related code constructions.
The rest of the paper, after the Introduction, is organized as follows.
Hermitian-Lifted Codes are introduced in Section \ref{sec:construction}, and information about their recovery sets is found there. The main theorem, providing the lower bound on the code rate, and proof are given in Section \ref{sec:mainthm}. Examples are given Section \ref{sec:examples}, and Section \ref{sec:conclusion} provides a brief conclusion. 

\subsection{The Hermitian Curve and Prior Code Constructions}

An algebraic geometric perspective has proven useful in coding theory, particularly in the case of evaluation codes using the points of curves (dimension 1 varieties) defined over finite fields.  These codes can be viewed as generalizations of Reed-Solomon codes, with the advantage that the length of the code is not bounded above by the field size, allowing for much longer codes with more codewords than a RS code over the same field.  The length of the evaluation code is bounded by the number of points on the curve over the given field; thus curves with as many points as possible over a given field are useful in this context. The Hermitian curve is extremal among maximal curves and is almost certainly the best-studied maximal curve of positive genus.

\subsubsection{The Hermitian curve}

For a curve $\mathcal{X}$ defined over a field $k$ and $K/k$ any field extension, let $\mathcal{X}(K)$ denote the set of points on $\mathcal{X}$ defined over $K$.  A curve's genus is a non-negative integer that is one measure of its complexity. The number of points possible for a curve over a finite field is limited by the Hasse-Weil bound by of the curve's genus and the field size as follows: for a smooth, projective $\mathcal{X}$ of genus $g$ defined over a finite field $\mathbb{F}_q$ of cardinality $q$, we have that 
\[q+1-2g\sqrt{q}\leq |\mathcal{X}(\mathbb{F}_q)|\leq q+1+2g\sqrt{q}.\]  A curve which obtains the upper bound over a given field is said to be \textit{maximal} over that field.

The Hermitian curve $\mathcal{H}_q$ is defined over $\mathbb{F}_q$ by the affine equation
\[ x^q +x =y^{q+1}.\]  This curve is smooth, irreducible, has genus $g=\frac{q(q-1)}{2}$, and has a single point at infinity, denoted by $P_{\infty}$.

Let $k$ be any natural number. We consider the points on $\mathcal{H}_q$ over the corresponding degree $k$ extension field of $\mathbb{F}_q$, given explicitly by
\[\mathcal{H}_q(\mathbb{F}_{q^k})=\{(x,y)\in \left(\mathbb{F}_{q^k}\right)^2 : x^q +x =y^{q+1}\}\cup\{P_{\infty}\}.\]  In this paper, we focus on the field $\mathbb{F}_{q^2}$.  Note that $\mathcal{H}_q$ has $q^3+1$ points over  $\mathbb{F}_{q^2}$, so $\mathcal{H}_q$ is maximal over $\mathbb{F}_{q^2}$.  

The Hermitian curve $\mathcal{H}_q$ is extremal in many ways: it is the unique curve with the largest possible genus for a curve maximal over the field $\mathbb{F}_{q^2}$, and thus is the maximal curve with the largest number of points for that field.  The curve is also as symmetrical as possible in that the automorphism group of $H_q$ is PGU($3, q^2$), which makes $\mathcal{H}_q$ the only curve of genus $g$ with automorphism group of order greater than $16g^4$ \cite{Stichtenoth1973}.  

This exceptional symmetry is apparent in the geometry of $\mathcal{H}_q$. The intersection of $\mathcal{H}_q$ with lines in the projective plane $\mathbb{P}^2$ will be very important to our construction. \begin{fact}\label{fact:lines}
Every line in $\mathbb{P}^2$ that is not tangent to $\mathcal{H}_q$ intersects $\mathcal{H}_q$ in exactly $q+1$ distinct places. Tangent lines to $\mathcal{H}_q$ intersect $\mathcal{H}_q$ in exactly one place \cite{HKT08}.
\end{fact}

We consider only lines defined over $\mathbb{F}_{q^2}$ which do not pass through the point $P_{\infty}$, which can be parameterized by the affine equations $x=\alpha t + \beta$ and $y=t$ for $\alpha,\beta\in\mathbb{F}_{q^2}$.  Note that each line tangent to $\mathcal{H}_q$ at an affine point does not pass through $P_{\infty}$, so such a line is of the given form. Thus, for each $P$ a point of $\mathcal{H}_q(\mathbb{F}_{q^2})\setminus \{P_{\infty}\}$, there are exactly $q^2 -1$ distinct, non-tangent lines passing through $P$ and $q$ other affine points of the curve defined over $\mathbb{F}_{q^2}$.

\subsubsection{Evaluation codes and one-point codes on Hermitian curves}
V.D. Goppa first defined evaluation codes on curves over finite fields in the early 1980s \cite{Goppa}.  The basic idea is to choose a set of points on the curve $\mathcal{X}$ as evaluation points, and a disjoint set of points as the support of a pole divisor.  Codewords are created by evaluating functions which take on poles only in the support of the pole divisor on the evaluation points.  The simplest case of Goppa's construction is a \textit{one-point code}, where the pole divisor is $D=mP$ for some natural number $m$ and $P$ a point on $\mathcal{X}$.  More concretely, we use the following definition.

\begin{definition} Let $\mathcal{X}$ be a smooth curve defined over $\mathbb{F}_q$.  Let $P$ be a point on $\mathcal{X}(\mathbb{F}_q)$ and $m$ be a natural number.  Let $B=\{P_1, P_2, \dots, P_n\}$ be a set of points in $\mathcal{X}(\mathbb{F}_q)$ not containing $P$, and let $D$ be the divisor $D:=P_1+P_2+\dots +P_n$.  Let $L(mP)$ be the Riemann-Roch space of functions on $\mathcal{X}$ with poles only at $P$ of order at most $m$.  The one-point code $C(D, mP)$ is the set $\{(f(P_1), f(P_2),\dots f(P_n))\in\left(\mathbb{F}_q\right)^n:f\in L(mP)\}$.
\end{definition}

For simplicity, we define a one-point code on the Hermitian curve to take $mP_{\infty}$ as the pole divisor.  These codes have been well-studied, beginning with work by Tiersma \cite{Tiersma} and Stichtenoth \cite{Stichtenoth_codes}.  The Riemann-Roch space $L(mP_{\infty})$ on $\mathcal{H}_q$ can be explicitly written down with basis
\[\{x^iy^j:0\leq j\leq q-1, iq + j(q+1)\leq m\}.\] We use evaluation set $B=\mathcal{H}_q(\mathbb{F}_{q^2})\setminus \{P_{\infty}\}$, to obtain the evaluation divisor $D$.  Adapting the notation of \cite{BallicoRavagnani}, we define the code $C_{q,m}$ to be one-point code $C(D,mP_{\infty})$ with choices as above.

The length of $C_{q,m}$ is $n=q^3$.  The dimension $k$ of $C_{q,m}$ is given by $\dim(L(mP_{\infty}))$ for $m<q^3$ and $k$ can be determined using the Riemann-Roch theorem.  If $m>2g-2$ we have $k=\dim(L(mP_{\infty}))=m-g+1$; in general, $k=\dim(L(mP_{\infty})) \geq m-g+1$. The minimum distance $d$ of the code can be bounded by $d\geq n-m$, since any function with a single pole of at most order $m$ can have at most $m$ zeros.  The exact minimum distance has been determined for all values of $m$ \cite{Stichtenoth2009, YangKumar}.

Evaluation codes with locality from algebraic curves appear in \cite{BargTamoVladut}.  The authors define locally recoverable codes on curves, with locality arising from covering maps and recovery based on polynomial interpolation, and define a locally recoverable code with availability $t=2$ on $\mathcal{H}_q$ by viewing the curve as a fiber product.  In \cite{HaymakerMalmskogMatthews}, the fiber product construction is utilized to define codes on curves with higher availability, with arbitrarily large availability possible for codes over large fields.  A code of length $q^3$ on $\mathcal{H}_q$ is defined with availability $t=\frac{\log q}{\log p}$, where $p$ is any prime and $q=p^{t}$.  In this paper, we use an entirely different approach define codes of the same length on $\mathcal{H}_q$ (with $q$ even) with vastly higher availability $t=q^2-1$.  

\subsubsection{Contrast of \ourcodes\ with related literature}
\ourcodes\ have some similarities with constructions in the literature, but are distinct in a few important ways. The locally recoverable codes in \cite{BargTamoVladut} and \cite{HaymakerMalmskogMatthews} are based on fiber products of curves, an algebraic geometry construction which builds a curve $\mathcal{Y}$ by the product of several other curves $\mathcal{Y}_i$, $1\leq i\leq t$, each with maps to a shared base curve $\mathcal{X}$.  The functions evaluated in these codes are multivariate polynomials of bounded degree in each of the generators of the corresponding extensions of function fields. The $t$ disjoint recovery sets in these codes correspond to fibers of the induced covering maps from $\mathcal{Y}$ to $\mathcal{Y}_i$.  In \ourcodes, the recovery sets correspond to the intersection points of the Hermitian curve with non-tangent lines.  These could be viewed as the fibers of projection maps along each non-horizontal slope from $\mathcal{H}_q$ to a copy of the projective line $\mathbb{P}^1$, but these projection maps are not induced by a fiber product construction and there is not an obvious construction of the functions evaluated in \ourcodes\ as any simple class of functions in the compositum of the function fields of the projective lines. For these codes to arise from the fiber product construction, we would need to construct the Hermitian curve $\mathcal{H}_q$ as a fiber product of $q^2$ projective lines.  Each map $\mathcal{Y}_i\rightarrow\mathcal{X}$ would need to be degree $q+1$, and could ramify above at most two points.  By counting points, we see this could not lead to $\mathcal{H}_q$.  Thus our construction creates a code which could not arise from the fiber product approach.

We also contrast \ourcodes\ with other constructions based on lifted codes.  Lifted codes and their variations~\cite{GuoKS13,BGKKS13,G15,Wu15,LW18,FischerGW17} have provided constructions of codes with good locality and availability.  It might be tempting to think that \ourcodes\ are identical to bivariate lifts of Reed-Solomon codes, punctured to the Hermitian curve, but this is not the case.  Indeed, \ourcodes\ correspond to bivariate polynomials over $\F_{q^2}$ whose restriction to every line \em intersected \em with the Hermitian curve have degree at most $q-1$.  The relevant lifted code---corresponding to bivariate polynomials whose restriction to every line has degree at most $q-1$---is known to be equal to the bivariate Reed-Muller code of degree $q-1$~\cite{RS96}.  In particular, this code has dimension $O(q^2)$ and in particular the puncturing to the Hermitian code (of size $q^3$) has rate $O(\frac{1}{q})$.  Thus, the \ourcode\ is much larger than the corresponding lifted Reed-Solomon code, punctured to the Hermitian curve.   We also note that \ourcodes\ are different from the notions of lifted Hermitian codes given in \cite{G15,BGKKS13}.
The main difference is that in those works, the Hermitian code can be seen as the ``base code,'' while in our work, the Hermitian code is used in the definition of the lifting process.


\section{Code Construction}\label{sec:construction}
In this section, we give a few preliminary definitions and define \ourcodes. For the rest of the paper, let $\mathcal{X}=\mathcal{H}_q(\mathbb{F}_{q^2})\setminus\{P_{\infty}\}$.

As discussed in the Introduction, \ourcodes \ are codes that are lifted with respect to Hermitian curves.  More precisely, a \ourcode \ is the evaluation code of all bivariate polynomials that agree with low-degree polynomials on all lines intersected with $\mathcal{X}$.  We formalize this below. 

First, we make the observation that one-point Hermitian codes are themselves naturally locally recoverable with locality $q$ and availability $q^2-1$.

\begin{observation} \label{obs_1pt}
The one-point code $C_{q,m}$ is locally recoverable with locality $q$ and availability $q^2-1$ for all $m\leq q^2-1$.
\end{observation}
\begin{proof}
Each index $i$ of a position in $C_{q,m}$ corresponds to a point $P_i$ in $\mathcal{H}_q(\mathbb{F}_{q^2})\setminus\{P_{\infty}\}$.  For any $\alpha,\beta\in\mathbb{F}_{q^2}$, we define the function $L_{\alpha,\beta}:\mathbb{F}_{q^2}\rightarrow \left(\mathbb{F}_{q^2}\right)^2$ so that \[L_{\alpha,\beta}(t)=(\alpha t+\beta, t).\]  For each such line $\Im(L_{\alpha,\beta})$ passing through $P_i$ which is not tangent to $\mathcal{H}_q$ at $P_i$, let $R_{i,\alpha}$ be the set of indices corresponding to points in the set $(\mathcal{H}_q\left(\mathbb{F}_{q^2})\cap \Im(L_{\alpha,\beta})\right)\setminus\{P_i\}$.  We see that $|R_{i,\alpha}|=q$ and there are $q^2-1$ such mutually disjoint sets for each $i$.  

Any codeword in $C_{q,m}$ is the evaluation of a function $f(x,y)$ which is a $\mathbb{F}_{q^2}$-linear combination of monomials of the form $x^ay^b$ where $b\leq q-1$ and $aq+b(q+1)\leq q^2-1$.  Thus 
\[a+b+b\frac{1}{q}\leq q-\frac{1}{q},\] so 
\[a+b\leq q-\frac{b+1}{q}.\]
Since $a$ and $b$ are non-negative integers, for each monomial in $f(x,y)$ we have $a+b\leq q-1$.  Thus for all points on the line $\Im(L_{\alpha,\beta})$, the function $f$ is identical to a univariate polynomial $g_{\alpha}(t)$ of degree at most $q-1$.  

If the symbol in position $i$ of the codeword corresponding to $f$ is erased, the value of $f(P_i)$ may be recovered by interpolating the polynomial $g_{\alpha}(t)$ from its values on the $q$ points with indices in $R_{i,\alpha}$.  
\end{proof}

It is worth noting that the one-point codes considered in Observation \ref{obs_1pt} have rate $$\frac{m+1-g+\dim L(K-mP_{\infty})}{q^3} \leq \frac{m+1}{q^3} \rightarrow 0$$ as $q \rightarrow \infty$, where $K$ denotes a canonical divisor on $\mathcal{H}_q$. In what follows, we develop \ourcodes, which have rate bounded away from $0$ as $q \rightarrow \infty$.

\begin{definition} For polynomials $f \in \F_{q^2}[x,y]$ and $g \in \F_{q^2}[t]$, and for a function $L:\F_{q^2}[t] \to \F_{q^2}^2$, we say that $f \circ L$ \em agrees with $g$ on $\mathcal{X}$ \em if $f(L(t)) = g(t)$ for all $t \in \F_{q^2}$ with $L(t) \in \mathcal{X}$.
\end{definition}

\begin{definition}\label{def:F}
Given a prime power $q$, 
let 
\[ \mathcal{F} = \inset{ f \in \F_{q^2}[x,y] \suchthat \begin{minipage}{8cm} \begin{center} $\forall L \in \mathcal{L}, \exists g \in \F_{q^2}[t]$  so that  $\deg(g) \leq q-1$  and so that $f \circ L$ agrees with $g$ on $\mathcal{X}$. \end{center} \end{minipage}} , \]
where above 
\[ \mathcal{L} = \inset{ L_{\alpha,\beta} \suchthat \alpha \in \F_{q^2}, \beta \in \F_{q^2} } \]
is the set of all lines of the form $L_{\alpha,\beta}(t) = (\alpha t + \beta, t)$.
\end{definition}

\begin{definition}[\ourcodes]\label{def:C}
Let $q$ be a prime power and let $\mathcal{F}$ be as in Definition~\ref{def:F}.
Define the \em \ourcode \em \ $\cC \subseteq \left(\F_{q^2}\right)^{q^3}$ as the evaluation code
\[  \cC = \inset{ \left( f(x,y) \right)_{(x,y) \in \mathcal{X}} \suchthat f \in \mathcal{F} }. \]
\end{definition}

We note that $\mathcal{X}$, $\mathcal{F}$ and $\cC$ depend on the choice of $q$; we suppress this in the notation since it will be clear from context. It is evident that $C_{q,m}$ is a subcode of $\cC$. 

\begin{remark}[Horizontal lines]
We ignore horizontal lines in our definition of $\mathcal{L}$ because they only intersect $\mathcal{X}$ in $q$ affine places, rather than $q+1$.  In particular, horizontal lines are different than non-horizontal lines because the point at $\infty$ is not an evaluation point in the code construction. As we see below, this will affect the locality of our resulting code.  
\end{remark}

It is easy to see (Observation~\ref{obs:local} below) that \ourcodes\ have locality and availability; the challenging task is to analyze the rate.

\begin{observation}\label{obs:local}
Let $q$ be any prime power, and 
let $\cC$ be the \ourcode \ as defined in Definition~\ref{def:C}.  Then
$\cC$ has locality $q$ and availability $q^2-1$.
\end{observation}
%
\begin{proof}
For any point $(x,y) \in \mathcal{X}$, there are $q^2 - 1$ lines $L_{\alpha,\beta}(t) \in \mathcal{L}$
that pass through $(x,y)$, that are not tangent to $\mathcal{X}$, and that are not horizontal.  
Any two of these lines intersect only in the point $(x,y)$, 
and each has $q$ points on $\mathcal{X}$ other than $(x,y)$.
These $q$ points form a repair group for the coordinate of $\cC$ indexed by $(x,y)$.
Indeed, let $f \in \mathcal{F}$, and suppose that $L(t)$ is such a line.  Let $t_0 \in \F_{q^2}$ be so that $L(t_0) = (x,y)$. As $f \in \mathcal{F}$, let $g(t)$ be a polynomial of degree at most $q-1$ so that $f(L(t)) = g(t)$ for any $t$ so that  $L(t) \in \mathcal{X}$.  
Given the $q$ values
\[ \inset{ f(x,y) \suchthat (x,y) \in (\Im(L) \cap \mathcal{X}) \setminus \{(x,y)\} }
= \inset{ g(t) \suchthat t \in \F_{q^2} \setminus \{t_0\} }. \]
of $f(x,y)$ on $(\Im(L) \cap \mathcal{X}) \setminus \{(x,y)\}$, one can use Lagrange interpolation to recover the polynomial $g$, and hence $g(t_0) = f(L(t_0)) = f(x,y)$.  Thus the symbol $f(x,y)$ of the codeword corresponding to $f$ can be recovered by the $q$ other symbols in the repair group corresponding to points in $(\Im(L) \cap \mathcal{X}) \setminus \{(x,y)\}$.
\end{proof}

\section{Main Theorem and Proof}\label{sec:mainthm}

Our main result is that \ourcodes \ have rate bounded below by a constant independent of $q$. It is an immediate observation that $\cC$ has rate at least $\frac{q(q+1)}{2q^3} \geq \frac{1}{2q}$, since $C_{q,q^2-1}$ is a subcode of $\cC$ and the dimension of $C_{q,q^2-1}$ is $\frac{q(q+1)}{2}$. 
However, what may be surprising is that in fact there are many functions $f \in \mathcal{F} \setminus L\left((q^2-1)P_{\infty}\right)$, enough so that the rate of the code $\cC$ is actually bounded below by a constant independent of $q$.

\begin{theorem}\label{thm:main}
Suppose that $q \geq 4$ is a power of $2$, and let $\cC$ be as in Definition~\ref{def:C}.  Then the rate of $\cC$ is at least $0.007$. 
\end{theorem}

For the rest of this section, we will assume that $q = 2^\ell$ is a power of two, as in the hypotheses of Theorem~\ref{thm:main}.
The strategy will be to find a large set of monomials $M_{a,b}(x,y) := x^ay^b$ for $a \leq q-1$ and $b \leq q^2 - 1$ so that $x^a y^b \in \cC$.  

It is not hard to see that such monomials lead to linearly independent codewords as shown in the next result.

\begin{proposition}\label{prop:linind}
Let $M_{a,b}(x,y) = x^a y^b$.  Then the set of vectors
\[ \inset{ \left( M_{a,b}(x,y) \right)_{(x,y) \in \mathcal{X}} \suchthat 0\leq a \leq q-1,0\leq b \leq q^2 - 1 } \] 
are linearly independent.
\end{proposition}
\begin{proof}
The kernel of the evaluation map of the affine points of the Hermitian curve is generated by $y^{q+1}-x^q-x, y^{q^2}-y, x^{q^2}-x$. Under any monomial ordering where $y^{q+1}<x^{q}$, the polynomials $y^{q+1}-x^q-x, y^{q^2}-y$ are a Gr\"obner basis for the kernel. Hence the monomial set $M_{a,b}(x,y)$, $0 \leq a \leq q-1$, $0 \leq b \leq q^2-1$ can not contain any element from the kernel of the evaluation map, which implies the evaluations of $M_{a,b}$ are linearly independent.
\end{proof}
Since such monomials lead to linearly independent codewords by Proposition~\ref{prop:linind}, bounding the number of them in $\cC$ will give us a lower bound on the dimension of $\cC$.

The proof proceeds in two steps.  We give a brief overview below, after we introduce some necessary notation.

\newcommand{\pab}{p_{\alpha, \beta}}
\newcommand{\dab}{\deg_{\alpha, \beta}}
\begin{definition}[$\pab,\dab$]\label{def:pab}
Given $\alpha, \beta \in \F_{q^2}$, define
\begin{equation}\label{eq:defofp}
 \pab(t) := t^{q+1} + \alpha^q t^q + \alpha t + (\beta + \beta^q) = t^{q+1} + \alpha^q t^q+ \alpha t + \gamma, 
\end{equation}
where above we are defining $\gamma := \beta + \beta^q$.
For a polynomial $g(t) \in \F_{q^2}[t]$,
let $\bar{g}(t)$ be the remainder obtained when 
$g(t)$ is divided by $\pab(t)$, and define
\[ \dab(g) := \deg( \bar{g}_{\alpha,\beta}(t) ). \]
Notice that $\dab(g) \leq q$ for all $g \in \F_{q^2}[t]$.
\end{definition}

To see why Definition~\ref{def:pab} is relevant, consider a line $L_{\alpha,\beta}(t) = (\alpha t + \beta, t)$.
Notice that $M_{a,b} \circ L_{\alpha,\beta}$ agrees with a polynomial $g$ of degree strictly less than $q$ on $\mathcal{X}$ if any only if
\[ \dab(M_{a,b} \circ L_{\alpha,\beta}) < q. \]
Indeed, write
\[ (M_{a,b} \circ L_{\alpha, \beta})(t) = h(t)\pab(t) + g(t) \]
for some polynomial $g(t)$ of degree at most $q$.  
Then for any $t$ so that $L_{\alpha,\beta}(t) \in \mathcal{X}$, 
we have by definition that $t^{q+1} = (\alpha + \beta t)^q + \alpha + \beta t,$ or in other words that $\pab(t) = 0$.  Thus, $M_{a,b} \circ L_{\alpha, \beta}$ agrees with $g(t)$ on $\mathcal{X}$, and since there are $q+1$ such values of $t$, $g(t)$ is the unique polynomial of degree at most $q$ for which this is true.

We say that a monomial $M_{a,b}$ is \textbf{good} if for all lines $L_{\alpha,\beta} \in \mathcal{L}$, 
\[ \dab(M_{a,b} \circ L_{\alpha, \beta} ) < q. \]  
The reasoning above leads to the following observation.
\begin{observation}\label{obs:good}
If $M_{a,b}$ is good, then $M_{a,b} \in \mathcal{F}$.
\end{observation}
Thus, our goal will be to find a big set of good monomials.
Our approach proceeds in two steps.  In the first step (Section~\ref{ssec:monomials}), we give a condition for when the monomial $t^k$ has degree at most $q-1$ modulo $\pab(t)$.  In the second step (Section~\ref{ssec:rate}), we use this condition, along with Lucas' theorem, to show that there are many good monomials.

\subsection{Behavior of monomials $t^k$ modulo $\pab(t)$}\label{ssec:monomials}

In this section, we give a condition on $k$ for the monomial $t^k$ to be low-degree modulo $\pab(t)$ and prove Theorem~\ref{lem:tpower} at the end of this section after we develop the necessary ingredients.
Let $\alpha, \beta$ be elements of $\mathbb{F}_{q^2}$ such that the line $L_{\alpha,\beta}(t) = (\alpha t + \beta, t)$ is not tangent to the Hermitian curve $\mathcal{X}$. 
As $\alpha,\beta$ are fixed for the rest of this section, for notational convenience the polynomial $\pab(t)$ will be denoted by $p(t)$ and $L_{\alpha,\beta}(t)$ will be denoted by $L(t)$.
As in Definition~\ref{def:pab}, we let $\gamma := \beta + \beta^q$.  Notice that $\gamma \in \F_q$.

Let $\sigma_0, \ldots, \sigma_q$ be the roots of $p(t)$.  There are $q+1$ distinct roots of $p(t)$ because there are $q+1$ distinct points in $\Im(L) \cap \mathcal{X}$.
Thus
\[ p(t) = t^{q+1} + \alpha^q t^q + \alpha t + \gamma=(t-\sigma_0)\cdots(t-\sigma_q)=c_0t^{q+1}+c_1t^q+\cdots+ c_qt+c_{q+1}, \]
where $c_k= \sum_{S \subset \{0, \ldots, q\}, |S| = k} \prod_{\ell \in S} \sigma_\ell,$ for $k=0,\ldots,q.$ In particular we have
\begin{align}
c_0 &= 1 \label{eq:c0}\\
c_1 &= \sum_{i=0}^q \sigma_i = \alpha^q \label{eq:c1}\\
c_k &= 0 \ \ \forall 1 < k < q \label{eq:ck}\\
c_q &= \sum_{i=0}^q \frac{ \sigma_0 \cdots \sigma_q }{\sigma_i} = \alpha \label{eq:cq}\\
c_{q+1} &= \sigma_0 \cdots \sigma_q = \gamma. \label{eq:cq1}
\end{align}

For any $k\geq 0$ we define the element $\displaystyle P_k=\sum_{i=0}^{q}\sigma_i^k .$ 
We show below that the values $P_k$ provide a sufficient condition to guarantee  $\dab(t^k) < q$.
\begin{proposition}\label{04.10.20}
Let $q$ be a power of $2$ and let $\alpha, \beta \in \F_{q^2}$.
Then $P_{k+1}=\alpha^q P_k$ if and only if $\dab(t^k) < q$. 
\end{proposition}
\begin{proof}
Write
\[ t^k= g_k(t)p(t) + \bar{g}_k(t) \]
for some polynomial $g_k(t)$ so that the polynomial $\bar{g}_k(t)$ has degree at most $q$. Our goal is to show that $\deg(\bar{g}_k(t)) < q$ if and only if $P_{k+1}=\alpha^q P_k$.

As $\sigma_0, \ldots, \sigma_q$ are the roots of $p(t)$, we have $\bar{g}_k(\sigma_i)=\sigma_i^k.$ 
Thus, we know $q+1$ values of $\bar{g}_k$. Since $\bar{g}_k$ has degree less than $q$, we may use Lagrange interpolation to write
\[ \bar{g}_k(t) = \sum_{i=0}^q \sigma_i^k \prod_{\ell \neq i} \inparen{ \frac{ t - \sigma_\ell }{ \sigma_i - \sigma_\ell } }=
\left(\sum_{i=1}^q \sigma_i^k \prod_{\ell \neq i} \frac{1}{\sigma_i - \sigma_\ell}\right) t^q + r(t), \]
where $\deg(r) < q$. 
Since
\[p(t) = t^{q+1} + \alpha^q t^q + \alpha t + \gamma=(t-\sigma_0)\cdots(t-\sigma_q),\] 
taking the derivative of both sides yields 
\[  p^\prime(t) = t^{q} +  \alpha =\sum_{i=0}^q \prod_{\ell \neq i} (t - \sigma_\ell). \]
Thus, 
\[ \displaystyle p^\prime(\sigma_i) = \sigma_i^{q} +  \alpha =\prod_{\ell \neq i} (\sigma_i - \sigma_\ell).\]

Because $\sigma_i$ is a root of $p(t)$, we have $\sigma_i^{q+1} + \alpha^q \sigma_i^q + \alpha \sigma_i = \gamma$; hence,  
\[  \left(\sigma_i^q + \alpha \right) \left(\sigma_i + \alpha^q \right) = \alpha^{q+1}+\gamma. \]
Thus, we get
\[ \displaystyle \prod_{\ell \neq i} (\sigma_i - \sigma_\ell)=\frac{\alpha^{q+1}+\gamma}{\sigma_i + \alpha^q}.\]
As a consequence the coefficient
of $t^q$ in $\bar{g}_k(t)$ is given by
\[\displaystyle \sum_{i=1}^q \frac{\sigma_i^k\left(\sigma_i + \alpha^q\right)}{\alpha^{q+1}+\gamma}=\frac{P_{k+1}+\alpha^qP_{k}}{\alpha^{q+1}+\gamma}.\]
Thus, this coefficient is zero exactly when $P_{k+1} = \alpha^q P_k$, as desired.
\end{proof}

The goal now is to find $k$ such that $P_{k+1}=\alpha P_k.$  We begin with an observation about $P_k$ for $0 \leq k < q$.

\begin{lemma}\label{04.07.20}
Let $q$ be a power of $2$.
For $0 \leq k < q$, $\displaystyle P_k = \alpha^{qk}$ and $\displaystyle P_{kq} = \alpha^{k}.$
\end{lemma}
\begin{proof}
Since $q$ is even, we have $P_0=1.$ 
Take $1\leq k < q.$ Newton's identities imply that 
$$ k c_k = \sum_{i=1}^k (-1)^{i-1} c_{k-i} P_k, $$
and 
replacing the $c_i$ with the values given in \eqref{eq:c0}-\eqref{eq:cq1},
we see that for $0\leq k<q,$ $P_k = \alpha^q P_{k-1}.$ 
Thus $\displaystyle P_k = \alpha^{qk}.$

Because we are working over $\F_{q^2}$, 
\[ P_{kq}=\sum_{i=0}^{q}\sigma_i^{kq} =\left(\sum_{i=0}^{q}\sigma_i^k\right)^{q}=\left(\alpha^{qk}\right)^q=\alpha^{k},\]
which completes the proof.
\end{proof}

We recall the \em Kronecker product \em of two matrices.
\begin{definition}
Let $A=[a_{ij}]$ be an  $r\times s$ matrix and $B=[b_{ij}]$ an $m_1\times m_2$ matrix. The Kronecker product of $A$ and $B$ is the
$rm_1\times sm_2$ matrix that can be expressed in block form as
\[ A\otimes B=\left(\begin{array}{cccc} a_{11}B & a_{12}B & \cdots & a_{1s}B \\ a_{21}B & a_{22}B & \cdots & a_{2s}B \\
\vdots & \vdots &  & \vdots \\ a_{r1}B & a_{r2}B & \cdots & a_{rs}B \\ \end{array}\right).\]
\end{definition}

\begin{proposition}\label{04.11.20}
Assume $q=2^\ell.$ Then 
\[\displaystyle \left(\begin{array}{cccc} P_0 & P_{q} & \cdots & P_{(q-1)q}  \\P_1 & P_{q +1} & \cdots & P_{(q-1)q+1}
\\ \vdots & \vdots & &\vdots  \\ P_{q-1} & P_{2q-1} & \cdots& P_{q^2-1} \end{array}\right)=
\left(\begin{array}{c c} 1 & \alpha^{2^{\ell-1}} \\ \alpha^{(2^{\ell-1})q} & \gamma^{2^{\ell-1}} \end{array}\right)
\otimes
\cdots
\otimes
\left(\begin{array}{c c} 1 & \alpha^2 \\ \alpha^{2q} & \gamma^2 \end{array}\right)
\otimes
\left(\begin{array}{c c} 1 & \alpha \\ \alpha^{q} & \gamma \end{array}\right).\]
\end{proposition}
\begin{proof}
Denote by $\Gamma_q$ the matrix of the left side and by $\Gamma_q^\prime$ the matrix of the right side of the proposed equality.
For a root $\sigma$ of $p(t) = t^{q+1} + \alpha^q t^q + \alpha t + \gamma,$ 
$\sigma^k=\alpha^q\sigma^{k-1}+ \alpha \sigma^{k-q}+\gamma \sigma^{k-q-1}$ for $k\geq q+1.$
Thus we obtain that the $P_k$ values satisfy the recurrence relation
\begin{equation}\label{04.08.20}
 P_k = \alpha^q P_{k-1} + \alpha P_{k-q} + \gamma P_{k-q-1}. 
\end{equation}
As a consequence, the $(i,j)$ entry on the matrix $\Gamma_q$ depends on the $(i-1,j), (i,j-1), (i-1, j-1)$ entries of $\Gamma_q$.  This implies that
the matrix $\Gamma_q$ is fully determined by its first row and its first column.
It is clear that the first row of $\Gamma_q^\prime$ is $(1,\alpha,\ldots,\alpha^q)$ and the first column of $\Gamma_q^\prime$ is
$(1,\alpha^q,\ldots,\alpha^{q(q-1)})^T.$  
Moreover, 
Lemma~\ref{04.07.20} implies that the same is true for $\Gamma_q$; thus the first rows and first columns of $\Gamma_q$ and $\Gamma_q^\prime$
are the same. 
In order to show that $\Gamma_q = \Gamma_q^\prime$,
we just need to verify that matrix $\Gamma_q^\prime$ satisfies~\eqref{04.08.20}. 
It is equivalent to show that every $2 \times 2$ block $M$ of $\Gamma_q^\prime$ satisfies the relation
\begin{equation}\label{04.09.20}
M_{22}=\alpha^qM_{12}+\alpha M_{21}+\gamma M_{11}.
\end{equation}
We proceed by induction. 
It is clear that the matrix $\left(\begin{array}{c c} 1 & \alpha \\ \alpha^{q} & \gamma \end{array}\right)$ satisfies~\eqref{04.09.20}.
Let $i > 1$ and assume that every $2\times 2$ block of the matrix
\[\displaystyle
{B}=\left(\begin{array}{c c} 1 & \alpha^{2^{i-1}} \\ \alpha^{(2^{i-1})q} & \gamma^{2^{i-1}} \end{array}\right)
\otimes
\cdots
\otimes
\left(\begin{array}{c c} 1 & \alpha^2 \\ \alpha^{2q} & \gamma^2 \end{array}\right)
\otimes
\left(\begin{array}{c c} 1 & \alpha \\ \alpha^{q} & \gamma \end{array}\right)\]
satisfies~\eqref{04.09.20}. We will show that the matrix 
$\left(\begin{array}{c c} 1 & \alpha^{2^{i}} \\ \alpha^{(2^{i})q} & \gamma^{2^{i}} \end{array}\right)
\otimes {B}$ satisfies~\eqref{04.09.20}.
Observe that the first row, first column, last row and last column of {B} are as shown:
\begin{equation}\label{04.12.20}
\left(\begin{array}{cccccc}
1 & \alpha & \alpha^2 & \ldots & & \alpha^{2^i-1} \\
\alpha^q &  &  &  & & \alpha^{2^i-2}\gamma \\
\alpha^{2q} &  &  &  & & \alpha^{2^i-3}\gamma^2 \\
\vdots &  &  &  & & \vdots \\
 &  &  &  & &  \\
\alpha^{(2^i-1)q} & \alpha^{(2^i-2)q}\gamma & \alpha^{(2^i-3)q}\gamma^2 & \ldots & & \gamma^{2^i-1}\end{array}\right).
\end{equation}
Now it is straightforward to check that the matrix
$\left(\begin{array}{c c} 1 & \alpha^{2^{i}} \\ \alpha^{(2^{i})q} & \gamma^{2^{i}} \end{array}\right)
\otimes {B}=\left(\begin{array}{c|c}{B} & \alpha^{2^i}{B} \\\hline  \alpha^{(2^i)q}{B} & \gamma^{2^i}{B}\end{array}\right)$
satisfies the desired property (\ref{04.09.20}). Indeed, take any $2 \times 2$ block $M.$ If $M$ belongs to any of the four blocks
${B}, \alpha^{2^i}{B}, \alpha^{(2^i)q}{B}$ or $\gamma^{2^i}{B}$, then we are finished by induction.
Otherwise, we have the following five cases:

\begin{itemize}
\item[(i)]
When
$M$ intersects the blocks 
${B}$ and $\alpha^{2^i}{B},$ 
$M=\left(\begin{array}{cc}
\alpha^{2^i-j}\gamma^{j-1} & \alpha^{2^i} \alpha^{jq}\\
\alpha^{2^i-{(j+1)}}\gamma^{j} & \alpha^{2^i} \alpha^{(j+1)q}
\end{array}\right)$ for some $j$.

\item[(ii)] 
When
$M$ intersects the blocks 
${B}$ and $ \alpha^{(2^i)q}{B},$
$M=\left(\begin{array}{cc}
\alpha^{(2^i-j)q}\gamma^{j-1} & \alpha^{(2^i-(j+1))q}\gamma^{j} \\
\alpha^j \alpha^{(2^i)q}& \alpha^{j+1} \alpha^{(2^i)q}
\end{array}\right)$ for some $j$.

\item[(iii)] When $M$ intersects the blocks 
$\alpha^{2^i}{B}$ and $\gamma^{2^i}{B},$
$M=\left(\begin{array}{cc}
\alpha^{2^i}\alpha^{(2^i-j)q}\gamma^{j-1} & \alpha^{2^i}\alpha^{(2^i-(j+1))q}\gamma^{j} \\
\alpha^j \gamma^{2^i}& \alpha^{j+1} \gamma^{2^i}
\end{array}\right)$ for some $j$.

\item[(iv)] When $M$ intersects the blocks
$ \alpha^{(2^i)q}{B}$ and $\gamma^{2^i}{B},$
$M=\left(\begin{array}{cc}
\alpha^{(2^i)q}\alpha^{2^i-j}\gamma^{j-1} & \gamma^{2^i} \alpha^{jq} \\
\alpha^{(2^i)q}\alpha^{2^i-{(j+1)}}\gamma^{j} & \gamma^{2^i} \alpha^{(j+1)q}
\end{array}\right)$ for some $j$.

\item[(v)] When $M$ intersects the four blocks---that is, $M$ is the $2 \times 2$ matrix in the center---we have
$M=\left(\begin{array}{cc}\gamma^{2^i-1} & \alpha^{2^i}\alpha^{(2^i-1)q} \\\alpha^{(2^i)q} \alpha^{2^i-1} & \gamma^{2^i}\end{array}\right).$
\end{itemize}

It is not hard to check that in all five cases, we have
$M_{22}=\alpha^qM_{12}+\alpha M_{21}+\gamma M_{11}.$
\end{proof}

Finally, we are ready to prove Theorem~\ref{lem:tpower}, stated below, which provides a sufficient condition for $\dab(t^k) < q$.

\begin{theorem} \label{lem:tpower}
Assume $q=2^\ell.$ Let $0 \leq k < q^2$, and write $k = wq + z $ where  $z < q.$   Let $\alpha, \beta \in \F_{q^2}$.
Suppose that either $w=0$, or that there exists $1 \leq i \leq \ell$ such that $w \equiv 0 \mod 2^i$ and $z \not\equiv -1 \mod 2^i$.
Then $\dab(t^k) < q$.
\end{theorem}

\begin{proof}[Proof of Theorem~\ref{lem:tpower}]
Suppose that $k = wq + z$ as in the theorem statement.
By Proposition~\ref{04.10.20},
 we just need to check that $P_{k+1}=\alpha^q P_k.$ 
When $w=0,$ it is clear that $P_{k+1}=\alpha^q P_k$ because 
by Lemma~\ref{04.07.20}, for $0 \leq k \leq q$, $\displaystyle P_k = \alpha^{qk}.$ 

Suppose that there exists an $i$ so that $w \equiv 0 \mod 2^i$ and $z \not\equiv -1 \mod 2^i$. Then let
\[\displaystyle
{A}=\left(\begin{array}{c c} 1 & \alpha^{2^{\ell-1}} \\ \alpha^{(2^{\ell-1})q} & \gamma^{2^{\ell-1}} \end{array}\right)
\otimes
\cdots
\otimes
\left(\begin{array}{c c} 1 & \alpha^{2^{i}} \\ \alpha^{(2^{i})q} & \gamma^{2^{i}} \end{array}\right)\text{ and }
{B}=\left(\begin{array}{c c} 1 & \alpha^{2^{i-1}} \\ \alpha^{(2^{i-1})q} & \gamma^{2^{i-1}} \end{array}\right)
\otimes
\cdots
\otimes
\left(\begin{array}{c c} 1 & \alpha \\ \alpha^{q} & \gamma \end{array}\right),
\]
so that $A \in \F_{q^2}^{2^{\ell -i } \times 2^{\ell - i }}$ and $B \in \F_{q^2}^{2^i \times 2^i}$.
By Proposition~\ref{04.11.20}
\[\displaystyle \left(\begin{array}{cccc} P_0 & P_{q} & \cdots & P_{(q-1)q}  \\P_1 & P_{q +1} & \cdots & P_{(q-1)q+1}
\\ \vdots & \vdots & &\vdots  \\ P_{q-1} & P_{2q-1} & \cdots& P_{q^2-1} \end{array}\right)=
A\otimes B=\left(\begin{array}{cccc} a_{11}B & a_{12}B & \cdots & a_{1s}B \\ a_{21}B & a_{22}B & \cdots & a_{2s}B \\
\vdots & \vdots &  & \vdots \\ a_{s1}B & a_{s2}B & \cdots & a_{ss}B \\ \end{array}\right)\]
where $s = 2^{\ell -i }$. 
Suppose that $P^k$ lies in the block $a_{cd}B$ for some $c,d \in \{1, \ldots, 2^{\ell-i}\}$. 
The fact that $w \equiv 0 \mod 2^i$ means that the element $P_k$ is in the first column of the block $a_{cd}B.$
The fact that $z \not\equiv -1 \mod 2^i$
means that $P_k$ is not in the last row of the block $a_{cd}B.$  In particular, $P_{k+1}$ is also in the block $a_{cd}B$.  Because of the structure of the first column of $B$, shown in \eqref{04.12.20}, we have $P_{k+1} = \alpha^q P_k$. Thus by Proposition~\ref{04.10.20}, we have $\dab(t^k) < q$.
\end{proof}

\subsection{Bound on the rate of the code}\label{ssec:rate}
In this section we use Theorem~\ref{lem:tpower} in order to bound the rate of our code construction below, completing the proof of Theorem~\ref{thm:main}. 

As discussed above, the strategy will be to find a large set of monomials $M_{a,b}(x,y) = x^ay^b$ for $a \leq q-1$ and $b \leq q^2 - 1$ so that $M_{a,b}$ is good, and hence, by Observation~\ref{obs:good},  $M_{a,b} \in \mathcal{F}$.  Since such monomials lead to linearly independent codewords by Proposition~\ref{prop:linind}, this will give us a lower bound on the dimension of $\cC$.


%

If $a + b < q$, then clearly $M_{a,b}$ is good.  Indeed, in this case $\deg(M_{a,b} \circ L_{\alpha, \beta})<q$ for all $\alpha,\beta$, and so reducing modulo $\pab$ does not change this.

If $a + b \geq q$,
there are two mechanisms that contribute to $M_{a,b}(x,y)$ being good.  
To see this, we may expand $M_{a,b} \circ L_{\alpha, \beta}$ as follows:
\begin{equation}\label{eq:expansion}
(M_{a,b} \circ L_{\alpha,\beta})(t) = M_{a,b}( \alpha t + \beta, t ) = (\alpha t + \beta)^a t^b
= \sum_{j \leq a} {a \choose j} \alpha^j \beta^{a-j} t^{b+j}. 
\end{equation}
The first mechanism that can contribute to the goodness of $M_{a,b}$ is that the terms $t^{b+j}$ in
\eqref{eq:expansion} could have small degree mod $\pab(t)$, such as per Theorem~\ref{lem:tpower}.  The second mechanism is that the binomial coefficients ${a \choose j}$ could vanish mod $2$.  To understand this second mechanism, we will use Lucas' Theorem, stated below.

\begin{definition}
Let $a$ and $b$ be integers between $0$ and $2^d-1$, and let $\bin(a) \in \{0,1\}^d$ denote the binary expansion of $a$.
We say that $a$ \textbf{lies in the $2$-shadow of $b$}, denoted $a \leq_2 b$, if 
\[ \supp(\bin(a)) \subseteq \supp(\bin(b)). \]
\end{definition}

\begin{theorem}[Lucas]
Let $0 \leq a \leq b$ be integers.  Then ${b \choose a}$ is zero mod $2$ if and only if $a \nleq_2 b$, meaning $a$ does not lie in the $2$-shadow of $b$. 
\end{theorem}

Before continuing, we give an example to illustrate how both mechanisms come into play. 
\begin{example}
Let $q=4$ and consider $M_{2,8}(x,y) = x^2y^{8}$.  This is a high-degree polynomial on the curve $\mathcal{X}$.  However, on every line $L_{\alpha,\beta}(t) = (\alpha t + \beta, t)$,  we have
\begin{align*}
 M_{2,8}(L_{\alpha,\beta}(t)) &= (\alpha t + \beta)^2 t^{8} \\
&= (\alpha^2 t^2 + \beta^2) t^{8} \\
&= \alpha^2 t^{10} + \beta^2 t^{8}. 
\end{align*}
In the second line above when the cross-terms $\beta \alpha t + \beta \alpha t = 0$ canceled, Lucas' theorem was in action, as the binomial coefficient ${2 \choose 1}$ vanishes.
Now in the third line, we are left with the two monomials $t^{10}$ and $t^8$. 
By Theorem~\ref{lem:tpower}, both of these reduce to something of degree less than $q$.
Indeed, we have $10 = 2+2\cdot q.$ As $2\equiv 0 \mod 2$ and $2\not\equiv -1 \mod 2,$ we have $\dab(t^{10})  < q.$
We have that $8 = 0+2 \cdot q.$ As $0\equiv0 \mod 2$ and $2\not\equiv -1 \mod 2,$ we obtain $\dab(t^{8})  < q.$
We conclude that $\dab(M_{2,8}(L_{\alpha,\beta}(t)))  < q$, and hence $M_{2,8}$ is good.

Notice that both mechanisms were important here.  In particular,  if the binomial coefficient ${2 \choose 1}$ had not disappeared, we would be left with a term $t^9$.  One can check that $\dab(t^9) = 4$ for some $\alpha,\beta$, and so this would result in a $t^4$ term in $(M_{2,8} \circ L_{\alpha,\beta})(t) \mod \pab(t)$ and $M_{a,b}$ would not be good.
\end{example}

Finally, we can prove our main theorem, Theorem~\ref{thm:main}, which says that the rate of a \ourcode\ is bounded below by a positive constant.
\begin{proof}[Proof of Theorem~\ref{thm:main}]

As per Observation~\ref{obs:good}, we will come up with a large set of monomials $M_{a,b}$ that are good.  By Proposition~\ref{prop:linind}, the resulting codewords are linearly independent, and this will yield a lower bound on the dimension of $\cC$.

Let $q = 2^\ell$ as in the theorem statement.  Below, for an integer $x$, write $x = \sum_r x_r 2^r$, so that $x_r$ denotes the $r$'th least significant bit in the binary expansion of $x$.

\begin{claim}\label{cl:good}
Suppose that $a \leq q-1$ and $b \leq q^2 - 1$ satisfy the following properties.
\begin{itemize}
\item[(i)] $b = wq + b'$ for some $w < q$ and some $b' < 2^{\ell-1}$, so that $w \equiv 0 \mod 2^i$ for some $1 \leq i \leq \ell$;
\item[(ii)] $a < 2^{\ell-1}$;
\item[(iii)] there is some $0 \leq s \leq i-1$ so that  $a_s = b'_s = 0$. 
\end{itemize}
Then $M_{a,b}$ is good.
\end{claim}
\begin{proof}
Suppose that $a,b$ satisfy (i)-(iii).
Let $L_{\alpha,\beta}(t) = (\alpha t + \beta, t)$ be a line in $\mathcal{L}$ and write
\begin{equation}\label{eq:exp2}
 (M_{a,b} \circ L_{\alpha, \beta})(t) = \sum_{j \leq a} {a \choose j} \alpha^j \beta^{a - j} t^{j+b} 
= \sum_{j \leq_2 a} \alpha^j \beta^{a-j} t^{j+b} 
\end{equation}
using Lucas' theorem in the second equality.
Notice that for any $j \leq_2 a$, we have $j < 2^{\ell-1}$ and $j_s = 0$, using properties (ii) and (iii).
Then the only monomials that appear in \eqref{eq:exp2} are of the form $t^k$ where
$k = wq + b' + j$
for $w,b'$ as in (i) and for $j \leq_2 a$.
Let $i$ be as in (i), so that $w \equiv 0 \mod 2^i$.
We claim that $b' + j \not\equiv -1 \mod 2^i$.  
Indeed, we can write
\[ b' = 2^{s+1} b'' + b''' \qquad \text{ and } j = 2^{s+1} j'' + j''' \]
for some $b''',j''' < 2^s$, using the fact that $b'_s = j_s = 0$.
Note that there exists some $c \leq 2^i - 2^{s+1}$ so that 
\[2^{s+1}(b'' + j'') \equiv c \mod 2^i. \]
(Indeed, this is true for any integer multiple of $2^{s+1}$.)
Thus,
\[ b' + j \equiv c + b''' + j''' \mod 2^i. \]
Since $b''', j''' < 2^s$, we have
\[ c + b''' + j''' < (2^i - 2^{s+1}) + (2^{s+1} - 1) = 2^i - 1, \]
which means that $b' + j \not\equiv -1 \mod 2^i$, as claimed.

Thus, $k$ is of the form $k = wq + z$ (where $z = b' + j$) so that $w \equiv 0 \mod 2^i$ and $z \not\equiv -1 \mod 2^i$.  By Theorem~\ref{lem:tpower}, $\dab(t^k) < q$.  Since this is true for every power $t^k$ that appears in \eqref{eq:exp2}, $\dab(M_{a,b} \circ L_{\alpha, \beta}) < q$ for all $\alpha, \beta$, and hence $M_{a,b}$ is good.
\end{proof} 

Finally, we count the number of pairs $a,b$ meeting the description in Claim~\ref{cl:good}.  We iterate over all $s$, where we take $s$ to be the smallest index so that $a_s = b'_{s} = 0$.  
For a given $s$, there are $4^s - 3^s$ ways to assign the bits $a_0, \ldots, a_{s-1}$ and $b'_0, \ldots, b'_{s-1}$, since there are only $3^s$ ways to never have $a_r = b'_r = 0$ for any $0 \leq r \leq s-1$.  Then there are $4^{\ell - s - 2}$ ways to assign the bits $a_{s+1}, \ldots, a_{\ell-2}, b'_{s+1}, \ldots, b'_{\ell-2}$.  Finally, there are $2^{\ell-s-1}$ ways to assign the bits $w_{s+1}, \ldots, w_{\ell-1}$.  Notice that we will choose $w_0, \ldots, w_{s} = 0$, ensuring that $w \equiv 0 \mod 2^{s+1}$ and in particular $w \equiv 0 \mod 2^i$ for some $i > s$. 

Thus, the total number of monomials meeting the description in Claim~\ref{cl:good} is 
\begin{align*}
\sum_{s=0}^{\ell-1} \inparen{ 4^s - 3^s } 4^{\ell - s - 2} 2^{\ell - s - 1}
&= \frac{2^{3\ell}}{32} \sum_{s=0}^{\ell-1} \inparen{ 4^s - 3^s } 8^{-s} \\
&= \frac{2^{3\ell}}{32} \cdot \frac{2}{5} \inparen{1 + 4 \inparen{\frac{3}{8}}^\ell - 5 \inparen{\frac{1}{2}}^\ell } \\
&\geq 0.007 \cdot q^3
\end{align*}
using the fact that $q = 2^\ell$ and the assumption that $\ell \geq 2$.
Since the length of $\cC$ is $q^3 = |\mathcal{X}|$, this implies that the rate of $\cC$ is at least $0.007$.
\end{proof}

We note that Claim~\ref{cl:good} does not take into account all of the good monomials; in particular, $M_{a,b}$ with  $a = b = 2^{\ell-1} - 1$ is a good monomial that is not covered.  In the examples in Section \ref{sec:examples}, we see higher rates. 

We conclude this section about the rate of the code with a very loose bound on another parameter, namely the minimum distance.

\begin{proposition}
The minimum distance $d$ of the code $\cC$ is bounded by $q^2\leq d\leq q^3-q^2+1$.  
\end{proposition}
\begin{proof}
The upper bound given by the fact that $C_{q,q^2-1}$ is contained in $\cC$, and the minimum distance of $C_{q,q^2-1}$ is given in Theorem 5 of \cite{Stichtenoth_codes}.  The lower bound is based on our recovery procedure.  If $V$ is a codeword with a non-zero symbol in position $i$, this corresponds to a function $f_V$ which is non-zero on the point $P_i$.  Position $i$ has $q^2-1$ disjoint recovery sets, and for each recovery set, at least one symbol in $V$ must be non-zero (since the zero polynomial will be interpolated if all symbols in the recovery set positions are zero).  Thus any codeword with any non-zero symbol must have non-zero symbols in at least $q^2$ positions. \end{proof}

\section{Examples}\label{sec:examples}

In computed examples, the actual code $\cC$ has rate much higher than the asymptotic lower bound computed in Section \ref{ssec:rate}.  The following examples illustrate how much higher.

\subsection{The code $\cC$ when $q=4$}

When $q=4$, we work with the Hermitian curve $x^4+x=y^5$, which has 65 points over $\mathbb{F}_{16}$ including one point at infinity, giving a code of length $n=64$.  The code $\cC$ has dimension 13, with basis the set of monomials $x^ay^b$ where \[(a,b)\in\{(0, 0), (0, 1), (0, 2), (0, 3), (0, 8), (0, 10), (1, 0), (1, 1), (1, 2), (2, 0), (2, 1), (2, 8), (3, 0)\}.\] These exponent pairs are plotted in Figure \ref{4fig}.  In contrast, the comparable non-lifted one-point Hermitian code $C_{4,15}$ has dimension 10.  Thus the rate of $\cC$ is $\frac{13}{64}\approx 0.20$, while the rate of $C_{4,15}$ is $\frac{10}{64}\approx 0.16$.

\begin{figure}[h!]
\centering
  \includegraphics[width=.8\linewidth]{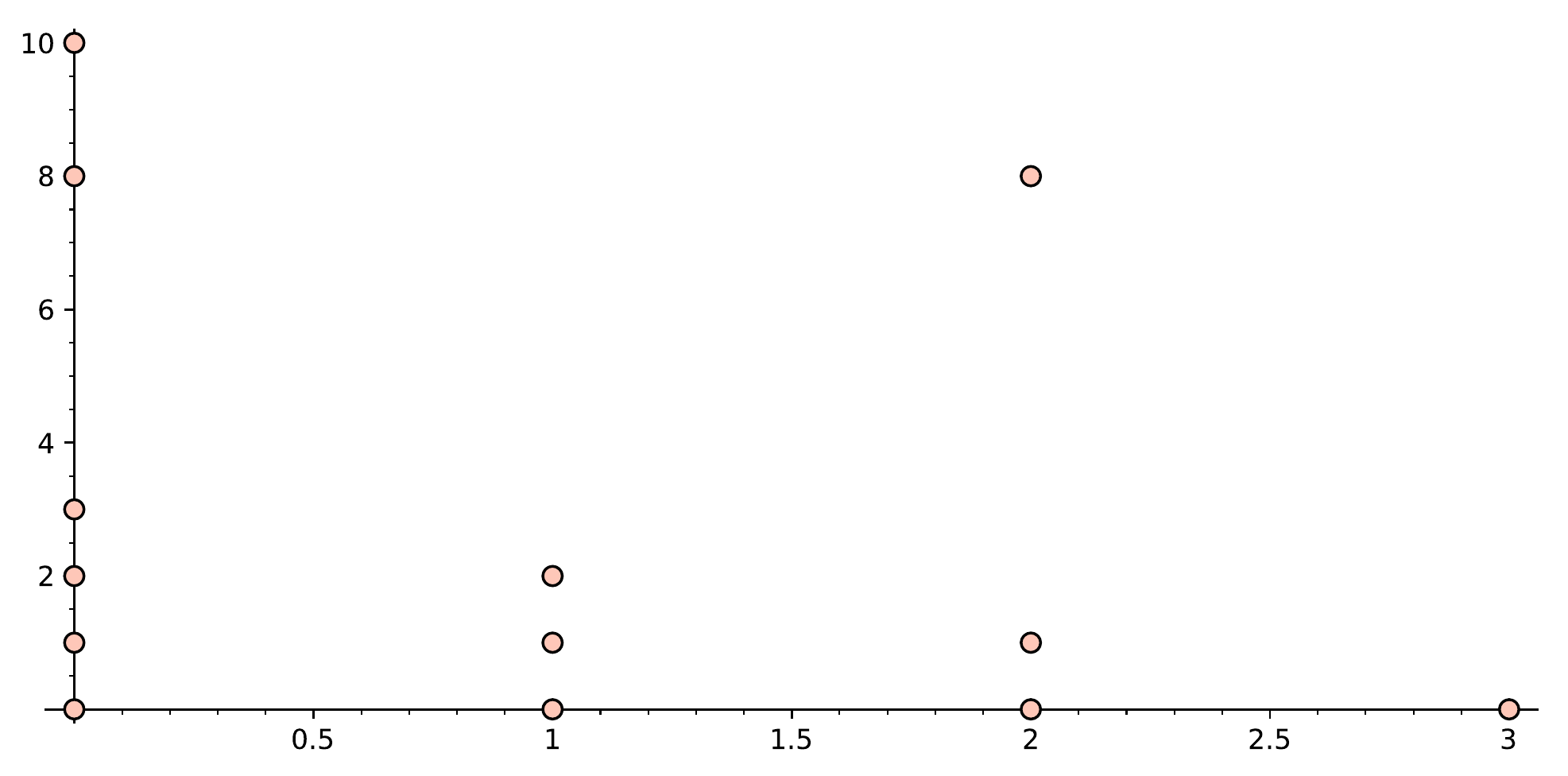}
  \caption{Exponent pairs $(a,b)$ with $x^ay^b\in \cC$ for $q=4$ ($a$ is on horizontal axis).}
  \label{4fig}
\end{figure}

\subsection{The code $\cC$ when $q=8$}

When $q=8$, we work with the Hermitian curve $x^8+x=y^9$, which has 513 points over $\mathbb{F}_{64}$ including one point at infinity, giving a code of length $n=512$.  The code $\cC$ has dimension 75, with basis the set of monomials $x^ay^b$ where 
\[(a,b)\in\{(0, 0), (0, 1), (0, 2), (0, 3), (0, 4), (0, 5), (0, 6), (0, 7), (0, 16), (0, 18), (0, 20), (0, 22), \]\[(0, 32), 
(0, 33), (0, 34), (0, 36), (0, 37), (0, 38), (0, 48), (0, 50), (0, 52), (0, 54), \] \[ (1, 0), (1, 1), (1, 2), (1, 3),
(1, 4), (1, 5), (1, 6), (1, 32), (1, 33), (1, 36), (1, 37),\]\[ (2, 0), (2, 1), (2, 2), (2, 3), (2, 4), (2, 5), (2, 16),
 (2, 18), (2, 20),\]\[ (2, 32), (2, 34), (2, 36), (2, 48), (2, 50), (2, 52),\] \[(3, 0), (3, 1), (3, 2), (3, 3), (3, 4), \] 
\[(4, 0), (4, 1), (4, 2), (4, 3), (4, 16), (4, 18), (4, 32), (4, 33), (4, 34), (4, 48), (4, 50),\]\[ (5, 0), (5, 1), 
(5, 2), (5, 32), (5, 33),\] \[(6, 0), (6, 1), (6, 16), (6, 32), (6, 48), (7, 0)\}.\]  These exponent pairs are plotted in Figure \ref{8fig}. In contrast, the comparable non-lifted one-point Hermitian code $C_{8,63}$ has dimension 36.  Thus the rate of $\cC$ is $\frac{75}{512}\approx 0.15$.  The rate of $C_{8,63}$ is $\frac{36}{512}\approx 0.07$.

\begin{figure}[h!]
\centering
  \includegraphics[width=.8\linewidth]{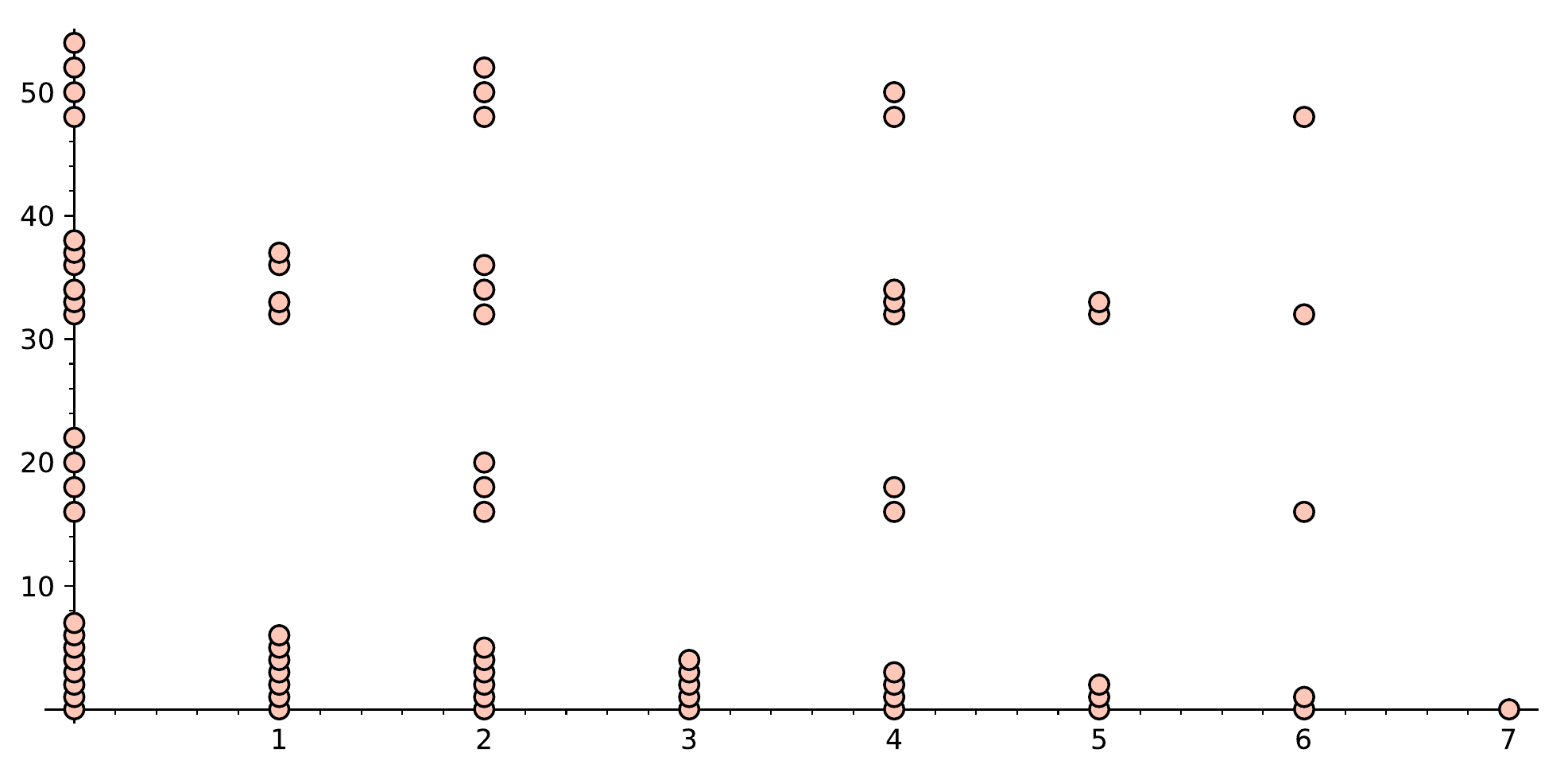}
  \caption{Exponent pairs $(a,b)$ with $x^ay^b\in \cC$ for $q=8$ ($a$ is on horizontal axis).}
  \label{8fig}
\end{figure}

\subsection{The code $\cC$ when $q=16$}

When $q=16$, we work with the Hermitian curve $x^{16}+x=y^{17}$, which has 4097 points over $\mathbb{F}_{256}$ including one point at infinity, giving a code of length $n=4096$.  The code $\cC$ has dimension 505, with basis the set of monomials $x^ay^b$ where $(a,b)$ are as depicted in Figure \ref{16fig}.  In contrast, the comparable non-lifted one-point Hermitian code $C_{16, 255}$ has dimension 136.  Thus the rate of $\cC$ is $\frac{505}{4096}\approx 0.123$.  The rate of $C_{16,255}$ is $\frac{136}{4096}\approx 0.033$.

\begin{figure}[h!]
\centering
  \includegraphics[width=.8\linewidth]{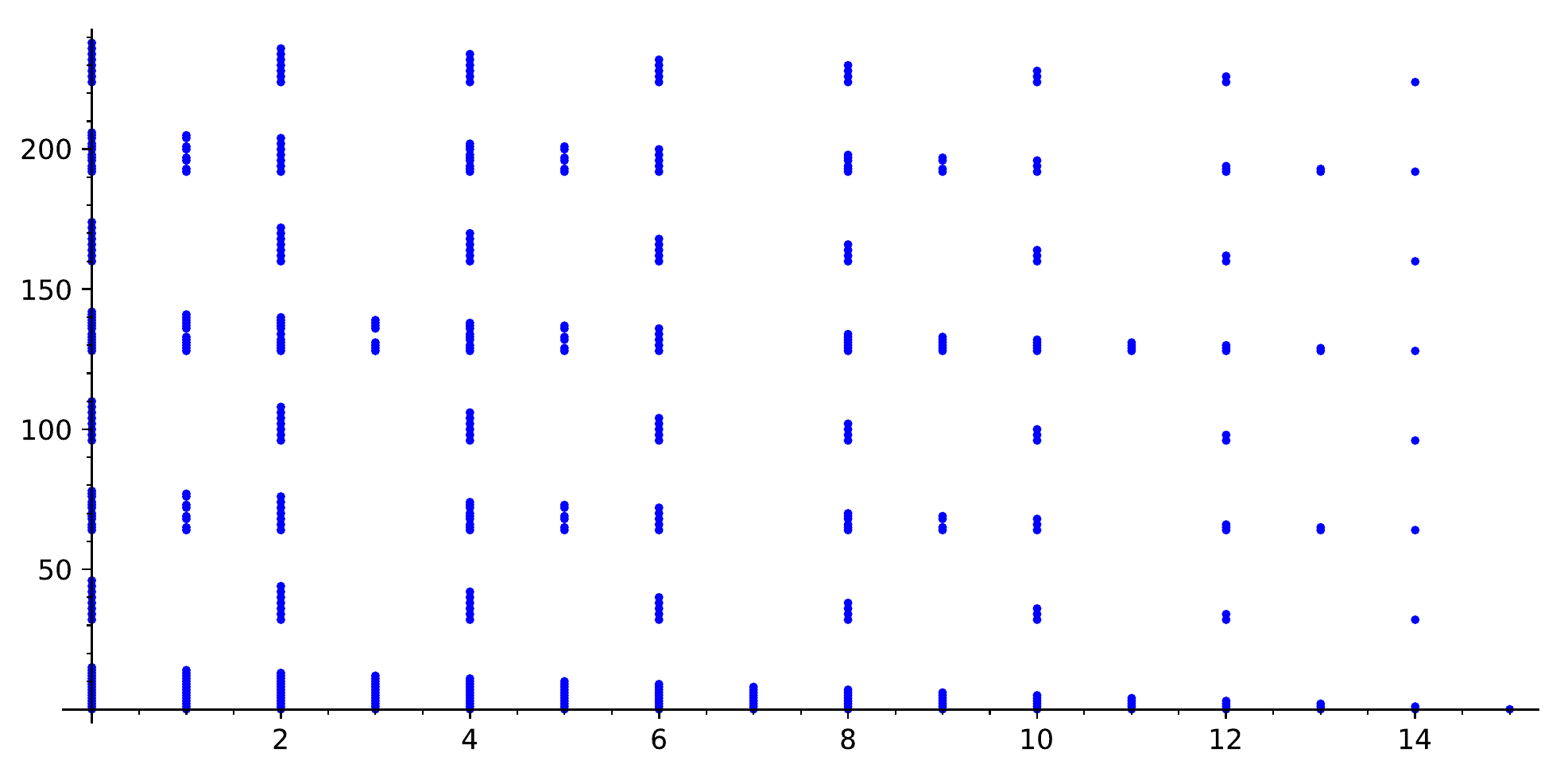}
  \caption{Exponent pairs $(a,b)$ with $x^ay^b\in \cC$ for $q=16$ ($a$ is on horizontal axis).}
  \label{16fig}
\end{figure}

\subsection{The code $\cC$ when $q=32$}

When $q=32$, we work with the Hermitian curve $x^{32}+x=y^{33}$, which has 32,769 points over $\mathbb{F}_{256}$ including one point at infinity, giving a code of length $n=32,768$.  The code $\cC$ has dimension 3675, with basis the set of monomials $x^ay^b$ where $(a,b)$ are as depicted in Figure \ref{32fig}.  In contrast, the comparable non-lifted one-point Hermitian code $C_{32, 1025}$ has dimension 528.  Thus the rate of $\cC$ is $\frac{3675}{32768}\approx 0.112$.  The rate of $C_{32,1025}$ is $\frac{528}{32768}\approx 0.016$.

\begin{figure}[h!]
\centering
  \includegraphics[width=.8\linewidth]{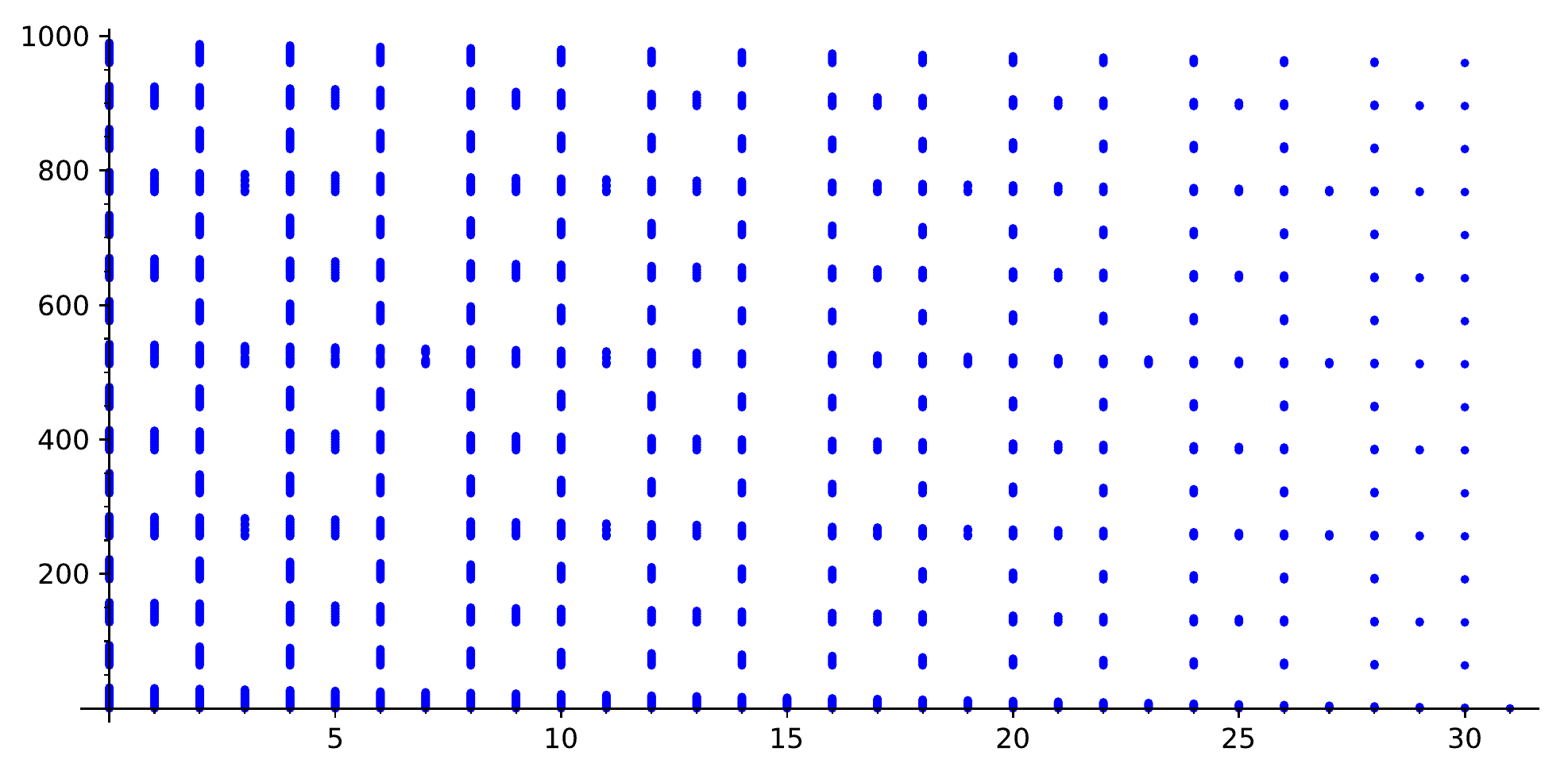}
  \caption{Exponent pairs $(a,b)$ with $x^ay^b\in \cC$ for $q=32$ ($a$ is on horizontal axis).}
  \label{32fig}
\end{figure}

\section{Conclusion}\label{sec:conclusion}

In this paper, we define \ourcodes, which are codes defined on the Hermitian curve
 with small locality, high availability, and rate bounded below by a constant. 
They are the evaluation code of polynomials whose restrictions to lines, intersected with the Hermitian curve, are all low-degree.  We study these codes as a first example of curve-lifted codes.  We establish the lower bound on the rate via a counting argument applied to certain ``good'' monomials.


We conclude with a few open questions.  First, it is an interesting question to completely characterize the ``good'' monomials for \ourcodes; determining their number would pin down the rate of these codes.
Second, it is interesting to explore other constructions of curve-lifted codes.  We view one of the main contributions of this work as introducing this paradigm for code constructions, and it is our hope that our construction and analysis of \ourcodes\ may serve as a prototype for  the construction and analysis of
other families of curve-lifted codes.

\section*{Acknowledgements} The authors thank AIM for hosting this collaboration through its SQuaREs program.

\bibliographystyle{plain}
\bibliography{refs}

\begin{thebibliography}{10}

\bibitem{BallicoRavagnani}
Edoardo Ballico and Alberto Ravagnani.
\newblock On the geometry of {H}ermitian one-point codes.
\newblock {\em Journal of Algebra}, 397:499--514, 2014.

\bibitem{BargTamoVladut}
Alexander Barg, Itzhak Tamo, and Serge Vl{\u{a}}du{\c{t}}.
\newblock Locally recoverable codes on algebraic curves.
\newblock {\em IEEE Transactions on Information Theory}, 63(8):4928--4939,
  2017.

\bibitem{BGKKS13}
Eli Ben-Sasson, Ariel Gabizon, Yohay Kaplan, Swastik Kopparty, and Shubangi
  Saraf.
\newblock A new family of locally correctable codes based on degree-lifted
  algebraic geometry codes.
\newblock In {\em Proceedings of the forty-fifth annual ACM symposium on Theory
  of Computing}, pages 833--842, 2013.

\bibitem{FischerGW17}
S.~Luna Frank{-}Fischer, Venkatesan Guruswami, and Mary Wootters.
\newblock Locality via partially lifted codes.
\newblock {\em CoRR}, abs/1704.08627, 2017.

\bibitem{Goppa}
V.~D. {Goppa}.
\newblock {Algebraico-geometric codes.}
\newblock {\em {Math. USSR, Izv.}}, 21:75--91, 1983.

\bibitem{G15}
Alan Guo.
\newblock High-rate locally correctable codes via lifting.
\newblock {\em IEEE Transactions on Information Theory}, 62(12):6672--6682,
  2015.

\bibitem{GuoKS13}
Alan Guo, Swastik Kopparty, and Madhu Sudan.
\newblock New affine-invariant codes from lifting.
\newblock In {\em Innovations in Theoretical Computer Science, {ITCS} '13,
  Berkeley, CA, USA, January 9-12, 2013}, pages 529--540, 2013.

\bibitem{HaymakerMalmskogMatthews}
Kathryn Haymaker, Beth Malmskog, and Gretchen~L Matthews.
\newblock Locally recoverable codes with availability t$\geq$2 from fiber
  products of curves.
\newblock {\em Advances in Mathematics of Communications}, 12(2):317, 2018.

\bibitem{HKT08}
James W.~P. Hirschfeld, G{\'a}bor Korchm{\'a}ros, and Fernando Torres.
\newblock {\em Algebraic Curves over a Finite Field}.
\newblock Princeton Series in Applied Mathematics. Princeton University Press,
  2008.

\bibitem{KT00}
Jonathan Katz and Luca Trevisan.
\newblock {On the efficiency of local decoding procedures for error-correcting
  codes}.
\newblock In {\em Proceedings of the 32nd symposium on Theory of Computing},
  STOC 2000, pages 80--86, 2000.

\bibitem{LW18}
Ray Li and Mary Wootters.
\newblock Improved list-decodability of random linear binary codes.
\newblock In {\em Approximation, Randomization, and Combinatorial Optimization.
  Algorithms and Techniques (APPROX/RANDOM)}. Schloss Dagstuhl-Leibniz-Zentrum
  fuer Informatik, 2018.

\bibitem{PV19}
Nikita Polyanskii and Ilya Vorobyev.
\newblock Trivariate lifted codes with disjoint repair groups.
\newblock In {\em 2019 XVI International Symposium on Problems of Redundancy in
  Information and Control Systems (REDUNDANCY)}, pages 64--68. IEEE, 2019.

\bibitem{RS96}
Ronitt Rubinfeld and Madhu Sudan.
\newblock Robust characterizations of polynomials with applications to program
  testing.
\newblock {\em SIAM Journal on Computing}, 25(2):252--271, 1996.

\bibitem{Ska16}
Vitaly Skachek.
\newblock Batch and {PIR} codes and their connections to locally repairable
  codes.
\newblock In {\em Network Coding and Subspace Designs}, pages 427--442.
  Springer, 2018.

\bibitem{Stichtenoth_codes}
H.~{Stichtenoth}.
\newblock A note on {H}ermitian codes over ${GF}(q^2)$.
\newblock {\em IEEE Transactions on Information Theory}, 34(5):1345--1348,
  1988.

\bibitem{Stichtenoth1973}
Henning Stichtenoth.
\newblock {\"U}ber die automorphismengruppe eines algebraischen
  funktionenk{\"o}rpers von primzahlcharakteristik.
\newblock {\em Archiv der Mathematik}, 24(1):527--544, 1973.

\bibitem{Stichtenoth2009}
Henning Stichtenoth.
\newblock {\em Algebraic function fields and codes}, volume 254.
\newblock Springer Science \& Business Media, 2009.

\bibitem{Tiersma}
H.~J. Tiersma.
\newblock Remarks on codes from {H}ermitian curves.
\newblock {\em IEEE Trans. Inf. Theor.}, 33(4):605–609, July 1987.

\bibitem{Woo10}
David~P. Woodruff.
\newblock {\em A Quadratic Lower Bound for Three-Query Linear Locally Decodable
  Codes over Any Field}, pages 766--779.
\newblock Springer Berlin Heidelberg, Berlin, Heidelberg, 2010.

\bibitem{Wu15}
Liyasi Wu.
\newblock Revisiting the multiplicity codes: A new class of high-rate locally
  correctable codes.
\newblock In {\em 2015 53rd Annual Allerton Conference on Communication,
  Control, and Computing (Allerton)}, pages 509--513. IEEE, 2015.

\bibitem{YangKumar}
Kyeongcheol Yang and P.~Vijay Kumar.
\newblock On the true minimum distance of {H}ermitian codes.
\newblock In {\em Coding theory and algebraic geometry}, pages 99--107.
  Springer, 1992.

\bibitem{yek_survey}
Sergey Yekhanin.
\newblock Locally decodable codes.
\newblock {\em Foundations and Trends{\textregistered} in Theoretical Computer
  Science}, 6(3):139--255, 2012.

\end{thebibliography}

\end{document}